\documentclass [6pt,a4paper]{article}
\usepackage{bbm}
\usepackage{amsfonts}
\usepackage {amssymb}
\usepackage {amsmath}
\usepackage {ntheorem}
\usepackage{latexsym}
\usepackage{booktabs}
\usepackage{bm}
\usepackage{multirow}
\usepackage[toc,page]{appendix}
\usepackage{indentfirst} 
\newenvironment{proof}{{\noindent\it\textbf{Proof}}\quad}{\hfill $\square$\par}
\usepackage{geometry}
\usepackage[figuresright]{rotating}
\def\dse#1{\vskip 0.6cm\noindent
        {\large\bf #1}
        \vskip 0.4cm}

\def\dse#1{\vskip 0.6cm\noindent
        {\large\bf #1}
        \vskip 0.4cm}

\usepackage{lastpage}
\usepackage{epsfig}

\begin{document}\large
\newtheorem{lemma}{Lemma}[section]
\newtheorem{theorem}[lemma]{Theorem}
\newtheorem{example}[lemma]{Example}
\newtheorem{definition}[lemma]{Definition}
\newtheorem{proposition}[lemma]{Proposition}
\newtheorem{conjecture}[lemma]{Conjecture}
\newtheorem{corollary}[lemma]{Corollary}
\newtheorem{remark}{Remark}
\begin{center}
\textbf{\LARGE{New entanglement-assisted quantum codes from negacyclic codes}}\footnote { E-mail addresses: chenxiaojing0909@ahu.edu.cn (X.Chen), xingbolu@yeah.net (X.Lu), zhushixin@hfut.edu.\\cn(S.Zhu), jiangw000@163.com(W.Jiang), xindi.wang@ahu.edu.cn (X.Wang).
}\\
\end{center}

\begin{center}
{ { Xiaojing Chen$^1$, Xingbo Lu$^1$, \ Shixin Zhu$^{2,3}$, \  Wan Jiang$^4$, \ Xindi Wang$^1$  } }
\end{center}

\begin{center}
\textit{1\  School of Internet, Anhui University, Hefei 230039, Anhui, P.R.China\\
2\ School of Mathematics, Hefei University of
Technology, Hefei 230601, Anhui, P.R.China\\
3\ Key Laboratory of Knowledge
Engineering with Big Data, Ministry
of Education, Hefei University of
Technology, Hefei 230601, Anhui, P.R.China\\
4\ School of Computer and Information, Hefei University of Technology, Hefei 230601, Anhui, P.R.China}
\end{center}

\noindent\textbf{Abstract:} The theory of entanglement-assisted quantum error-correcting codes (EAQECCs) is a generalization of the standard stabilizer quantum error-correcting codes, which can be possibly constructed from any classical codes by relaxing the duality condition and utilizing preshared entanglement between the sender and receiver. In this paper, a new family of EAQECCs is constructed from negacyclic codes of length $n=\frac{q^2+1}{a}$, where $q$ is an odd prime power, $a=\frac{m^2+1}{2}$ and $m$ is an odd integer. 
Some new entanglement-assisted quantum maximum distance separable (EAQMDS) codes are obtained in the sense that their parameters are not covered by the previously known ones.\\



\section{Introduction}

QECCs have been introduced as a powerful tool for protecting quantum information from decoherence and quantum noise. The theory of QECCs has experienced tremendous growth since the works of Steane \cite{ref50} and Shor \cite{ref51}. It was shown that the construction of QECCs can be reduced to that of classical linear codes satisfied with certain dual-containing condition \cite{ref41}. Until 2006, Brun $\it et~al.$ \cite{ref1}, a new quantum code named entanglement-assisted quantum error-correcting code (EAQECC), was introduced which overcomes the barrier of the dual-containing condition in constructing standard quantum codes from classical codes. Since then, much work on the constructions of EAQECCs from classical error-correcting codes has been done (see \cite{ref54,ref5,ref7,ref26,ref21,ref3,ref14,ref22,ref17,ref4,ref53}, and the relevant references therein).


It is well known that EAQECCs can be constructed from arbitrary classical codes by utilizing pre-shared entanglement between the sender and receiver. However, it is not an easy task to determine the number of entangled states $c$. Scholars have proposed several means to solve this problem. Decomposing the defining set is one of the solutions. L$\ddot{\textrm{u}}$ $\it et~al.$ \cite{ref10} constructed some EAQMDS codes by decomposing the defining set of cyclic codes, which transmitted the calculation of entangled states $c$ into determining a subset of the defining set of the underlying codes. Subsequently, Chen $\it et~al.$ \cite{ref11} and Chen $\it et~al.$ \cite{ref12} generalized this mean to negacyclic codes and constacyclic codes respectively. Then the theory of EAQECCs obtained tremendous development and many EAQMDS codes with certain numbers of entangled states $c$ [see \cite{ref45,ref55,ref19,ref20,ref57,ref58,ref77} and the relevant references therein] have been constructed.


Actually, the larger the
minimum distance of EAQECCs are, the more the entangled states $c$ will be employed.
However, it is difficult to analyze the accurate parameters if the value of
entangled states $c$ is too large or general. Recently, constructing EAQMDS codes with general parameters have obtained great growth \cite{ref56,ref32}. Some new EAQMDS codes of length $n=q^2+1$ and entanglement-assisted quantum almost maximum distance separable (MDS) codes of length $n=q^4-1$ with general parameters are constructed by Qian and Zhang in \cite{ref23}. Later, Wang $\it et~al.$ \cite{ref24} obtained several classes of EAQECCs of length $n=q^2+1$ based on constacyclic codes. Moreover, almost all of those known results about EAQECCs with the same length are some special cases of theirs. Very recently, Pang $\it et~al.$ \cite{ref43} gave some EAQECCs with general parameters of lengths $n=q^2+1$ and $n=\frac{q^2+1}{2}$ from negacyclic MDS codes and constacyclic MDS codes respectly. 
We note that most of their lengths are divided by $q^2+1$. Enlightened by their work, we constructed four families of EAQECCs from negacyclic codes with unified form and general parameters. Further, a large number of EAQMDS codes with minimum distance $d\leq\frac{n+2}{2}$
can be obtained. 
The parameters of EAQECCs are listed in Table 1, where $n=\frac{2(q^2+1)}{a}$, $q$ is an odd prime power, $a=m^2+1$ and $m$ is an odd integer.





This paper is organized as follows. In Sect.2, some basic background and results about negacyclic codes are reviewed. In Sect.3, we recall the relevant concepts and properties of EAQECCs. In Sect.4, four families of new EAQECCs and EAQMDS codes from negacyclic codes with unified form are presented. Sect.5 concludes the paper.

\begin{table}[h]
\setlength{\abovecaptionskip}{0.cm}
\setlength{\belowcaptionskip}{0.2cm}
\small
\centering
\caption{\newline  ~The length $n=\frac{2(q^2+1)}{a}$ parameters of EAQECCs.}
\begin{tabular}{c|c|c|c}
	\hline
	$[[n,k,d;c]]_q$ & $q$ &$m$ & $\xi$  	\\
	\hline
	$k = n-4\alpha q+4(a-m)(\alpha -\xi ) +2a\alpha^2-2a+4m-1$&$q=a\xi+a-m$& $m\equiv 1$ mod $4$ & odd \\
	&&  &  \\ 
	$d = 2[\alpha q+(a-m)\xi+a-2m +1 $&&  &  \\ 
	\cline{3-4}
    & & $m\equiv 3$ mod $4$ & even \\
	$c = 2\alpha[(a\alpha +2(a-m)]+2a-4m+1 $&&  &  \\ 
	&&  &  \\
	\hline
	
	$k = n-4\alpha q+4(a-m)(\alpha -\xi ) +2a\alpha^2-2a+4m-2$ & $q=a\xi+a-m$&$m\equiv 1$ mod $4$ & even \\
	&&  &  \\ 
	$d = 2[\alpha q+(a-m)\xi+a-2m +1] $&&  &  \\ 
	\cline{3-4}
	& & $m\equiv 3$ mod $4$ & odd \\
	$c = 2\alpha[(a\alpha +2(a-m)]+2a-4m $&&  &  \\ 
	&&  &  \\
	\hline
	
	$k = n-4\alpha q+4m(\alpha-\xi)+2a\alpha^2-1 $ & $q=a\xi+m$&$m\equiv 1$ mod $4$ & odd \\
	&&  &  \\ 
	$d = 2(\alpha q+m\xi +1) $&&  &  \\ 
	\cline{3-4}
	& & $m\equiv 3$ mod $4$ & even \\
	$c = 2\alpha(a\alpha +2m)+1 $&&  &  \\ 
	&&  &  \\
	\hline
	
    $k = n-4\alpha q+4m(\alpha-\xi)+2a\alpha^2-2 $ & $q=a\xi+m$&$m\equiv 1$ mod $4$ & even \\
	&&  &  \\ 
	$d = 2(\alpha q+m\xi +1) $&&  &  \\ 
	\cline{3-4}
	& & $m\equiv 3$ mod $4$ & odd \\
	$c = 2\alpha(a\alpha +2m) $&&  &  \\ 
	&&  &  \\
	\hline
				
\end{tabular}
\end{table} 

\section{Review of Negacyclic Codes} 
In this section, we will review some basic concepts and relevant results about negacyclic codes. 

Throughout this paper, let $\mathbb{F}_{q^{2}}$ be a finite field with $q^{2}$ elements and $\mathbb{F}_{q^{2}}^{n}$ be the $n$-dimensional row vector space over $\mathbb{F}_{q^{2}}$,
where $q$ is a power of prime $p$ and $n$ is a positive integer. A linear code $\mathcal{C}$ of length $n$ over $\mathbb{F}_{q^{2}}$ is a nonempty subspace of $\mathbb{F}_{q^{2}}^{n}$. A $q^2$-ary linear code $\mathcal{C}$ of length $n$ is negacyclic if $\mathcal{C}$ is invariant under the permutation of $\mathbb{F}_{q^{2}}^{n}$
$$(c_0,c_1,\ldots,c_{n-1})\rightarrow(-c_{n-1},c_0,\ldots,c_{n-2}).$$
We identify any codeword $\textbf{c}=(c_0,c_1,\ldots,c_{n-1})\in\mathcal{C}$ with its polynomial representation $c(x)=c_0+c_1x+\cdots+c_{n-1}x^{n-1}\in \mathbb{F}_{q^{2}}[x]/\langle x^{n}+1\rangle$.
Then $\mathcal{C}$ is a negacyclic code of length $n$ over $\mathbb{F}_{q^{2}}$ if and only if $\mathcal{C}$ is an ideal of the quotient ring $\mathbb{F}_{q^{2}}[x]/\langle x^{n}+1\rangle$. Moreover, every ideal of $\mathbb{F}_{q^{2}}[x]/\langle x^{n}+1\rangle$ is a principal ideal, so $\mathcal{C}$ can be generated by a monic divisor $g(x)$
of $x^{n}+1$. Let $\mathcal{C}=\langle g(x) \rangle$, where $g(x)$ is a unique monic polynomial which has minimal degree in
$\mathcal{C}$. Then $g(x)$ is called the generator polynomial of the code $\mathcal{C}$ and the dimension of $\mathcal{C}$
is $n-k$, where $k=\deg(g(x))$.

For any element $a\in\mathbb{F}_{q^{2}}$,
we use $\overline{a}$ to denote the conjugate $a^{q}$ of $a$. Given two vectors
$\mathbf{x}=(x_0,x_1,\ldots,x_{n-1})$ and $\mathbf{y}=(y_0,y_1,\ldots,y_{n-1})\in \mathbb{F}_{q^{2}}^{n}$,
their Hermitian inner product is defined as
\[ \langle \mathbf{x},\mathbf{y}\rangle=x_0\overline{y}_0+x_1\overline{y}_1+\cdots+x_{n-1}\overline{y}_{n-1}\in\mathbb{F}_{q^{2}}.\]
The vectors $\mathbf{x}$ and $\mathbf{y}$ are called orthogonal with respect to the Hermitian
inner product if $\langle \mathbf{x},\mathbf{y}\rangle=0$.
For a linear code $\mathcal{C}$, the Hermitian
dual code of $\mathcal{C}$ is defined as
\[\mathcal{C}^{\bot_{h}}=\{\mathbf{x}\in\mathbb{F}_{q^{2}}^{n}~|~\langle \mathbf{x},\mathbf{y}\rangle=0 ~for~ all~\mathbf{y}\in\mathcal{C}\}.\]
A linear code $\mathcal{C}$ of length $n$ is called Hermitian
self-orthogonal if $\mathcal{C}\subseteq\mathcal{C}^{\bot_{h}}$, and it is called Hermitian self-dual if
$\mathcal{C}=\mathcal{C}^{\bot_{h}}$.

Note that $x^{n}+1$ has no repeated root over $\mathbb{F}_{q^{2}}$ if and only if $\gcd(n,q)=1$. From now on this paper, assume that $\gcd(n,q)=1$. Let $\beta$ be a primitive $2n$-th root of unity in $\mathbb{F}_{q^{2m}}$, where $m$ is the multiplicative order of ${q^2}$ modulo $2n$. 
Let $\gamma=\beta^2$, then $\xi$ is a primitive
$n$-th root of unity. Hence, the root of $x^{n}+1$ is $\beta\gamma^j=\beta^{1+2j}, 0\leq j \leq n-1$.

Let $\mathbb{Z}_{2n}=\{0,1,\cdots,2n-1\}$ and $\varOmega_{2n}$ be the set of form $1+2j$ in $\mathbb{Z}_{2n}$. For any $i\in \mathbb{Z}_{2n}$, let $C_i$ be the
$q^2$-cyclotomic coset modulo $2n$ containing $i$,
\[C_i=\{i,iq^2,iq^4,\ldots,iq^{2(l_i-1)}\},\]
where $l_i$ is the smallest positive integer such that $iq^{2l_i}\equiv i \mod 2n$. In fact, $M_i(x)=\Pi_{t\in C_i}(x-\beta^t)$ is the minimal polynomial of $\beta^i$ over $\mathbb{F}_{q^2}$. Hence,  Each $C_i$
corresponds to an irreducible divisor of $x^{n}+1$ over $\mathbb{F}_{q^{2}}$. For a negacyclic code $\mathcal{C}=\langle g(x)\rangle$ of length $n$ over $\mathbb{F}_{q^2}$, its defining set is the set $Z=\{i\in\varOmega_{2n}~|~g(\beta^i)=0\}$.
Obviously, the defining set of $\mathcal{C}$ must be a union of some $q^2$-cyclotomic cosets modulo $2n$ and
dim$(\mathcal{C})=n-|Z|$, where $|Z|$ denotes the cardinality of the set $Z$. Next, we give the famous Singleton bound for linear codes first.

\begin{theorem} \label{th:2.1} \emph{\textbf{\cite{ref27} (Singleton bound)}} Let $\mathcal{C}$ be an $[n,k,d]$ linear code of length $n$ with dimension $k$ and minimum distance $d$, then $$n-k\geq d-1.$$ If the equality $n-k=d-1$ holds, then the code is an MDS code. 
\end{theorem}

From the property of pseudo-cyclic codes in \cite{ref29}, we can obtain the following BCH bound of negacyclic codes.

\begin{theorem} \label{th:2.2} \emph{\textbf{ \cite{ref29} (BCH bound for negacyclic codes)}} ~Assume that $n$ and $q$ are coprime.
	Let $\mathcal{C}$ be a negacyclic code of length $n$ over  $\mathbb{F}_{q^{2}}$, and let its
	generator polynomial $g(x)$ have the elements $\{\beta^{1+2j}~|~b\leq j\leq b+\delta-2\}$ as the roots, where
	$\beta$ is a primitive $2n$-th root of unity and $b$ is a integer. Then the minimum distance of $\mathcal{C}$ is at least $\delta$.
\end{theorem}

\section{Review of EAQECCs}

Let $q$ be a prime power. A $q$-ary $[[n, k, d; c]]$ entanglement-assisted quantum error-correcting code (EAQECC) is denoted by $[[n, k, d; c]]_q$, which encodes $k$ information qudits into $n$ channel qudits with the help of $c$ pairs of maximally entangled states and can correct up to $\lfloor \frac{d-1}{2}\rfloor$ errors, where qudits are quantum digits and $d$ is the minimum distance of the code. 

As we mentioned above, although it is possible to construct EAQECCs from any classical linear codes over $\mathbb{F}_{q^2}$, it is not easy to calculate the number of entangled states $c$. After the pioneering work of L$\ddot{\textrm{u}}$ $\it et~al.$ \cite{ref10} about decomposing defining set, some scholars have generalized their methods to more universal cases \cite{ref11,ref12}. Therefore, we can construct EAQECCs from negacyclic codes by using the following theorem.

\begin{theorem} \label{th:3.1}
	\emph{\textbf{ \cite{ref45}}} Let $\mathcal{C}$ be a negacyclic code of length $n$ with defining set $Z$ over $\mathbb{F}_{q^2}$.
	Suppose that $Z = Z_{1} \cup Z_{2}$ is a
	decomposition of $Z$, where $Z_{1} = Z \cap (-qZ)$, $Z_{2} = Z\backslash Z_{1}$ and $-qZ = \{2n - qx~|~x \in Z\}$. If $\mathcal{C}$ has parameters $[n, n-|Z|, d]_{q^2}$, then there exists an EAQECC with parameters $[[n, n-2|Z|+|Z_{1}|, d; |Z_{1}|]]_q$.
\end{theorem}

 As we all know, one of the central tasks in quantum coding theory is to construct quantum MDS codes which satisfy the quantum Singleton bound. Similar to the Singleton bound of quantum codes, there also is the following entanglement-assisted quantum Singleton bound  which established by Brun $\it et~al$ \cite{ref1}. It's worth pointing out that Grassl $\it et~al.$ \cite{ref37} gave a proof of entanglement-assisted quantum Singleton bound for arbitrary $q$. Therefore, we have the following results for any $q$.

\begin{theorem} \label{th:3.2}
	\emph{\textbf{ \cite{ref40}}} Assume that $\mathcal{C}$ is an entanglement-assisted quantum code with parameters $[[n, k, d; c]]_q$. If $d\leq \frac{n+2}{2} $, then $\mathcal{C}$ satisfies the entanglement-assisted Singleton bound \[n+c-k\geq 2(d-1).\] If $\mathcal{C}$ satisfies the equality $n+c-k= 2(d-1)$ for $d\leq \frac{n+2}{2} $, then it is called an entanglement-assisted quantum MDS code.
\end{theorem} 

\section{Construction of EAQMDS Codes with New Parameters}
\label{sec:3} 
The main idea of this section is giving a unified form of constructing EAQECCs of length $n$. Consequently, a family of new EAQMDS codes of length $n=\frac{2(q^2+1)}{a}$ with general parameters are presented, where $q$ is an odd prime power, $a=m^2+1$ and $m$ is an odd integer.

\begin{lemma} \label{le:4.1}
	~Let $n=\frac{2(q^2+1)}{a}$ and $s=\frac{n}{2}$, where $q$ is an odd prime power, $a=m^2+1$, and $m$ is an odd integer. Then, the $q^2$-cyclotomic cosets modulo $2n$ are $C_s=\{s\}$, $C_{3s}=\{3s\}$ and $C_{s-2l}=\{s-2l,s+2l\}$ for $1\leq l\leq s-1$.
\end{lemma} 

    Note that Lemma \ref{le:4.1} is a parallel promotion of Lemma 4.1 in \cite{ref44}, so we omit its proof here for simplification.\\
    
    Let $n=\frac{2(q^2+1)}{a}$, where $q$ is an odd prime power, $a=m^2+1$ and $2\mid a$, then $q\equiv \pm m $ mod $a$ and $m$ is an odd integer. Next, we will discuss all cases according to the different values of $q$ and $m$.
    
    In order to prove the later theorem more easily, we give another expression of the $q^2$-cyclotomic cosets modulo $2n$ containing integers in $\mathbb{Z}_{2n}$ based on the results of Lemma \ref{le:4.1}. We use $i$, $j$ and $q$ to replace the position of $l$. Taking into account all possible situations, we have given the value range of the corresponding parameters. The obtained results are listed below:
    $$C_{s+2(iq+j)}=\{s+2(iq+j), s-2(iq+j)\},$$

    (1) $q=a\xi+a-m$, the range of $i$, $j$ is shown in Table 2.\\
    \begin{table}[h]
    \setlength{\abovecaptionskip}{0.cm}
    \setlength{\belowcaptionskip}{0.2cm}
    \small
    \centering
    \caption{\newline  Parameters of $i$ and $j$.}
    \begin{tabular}{c|c}
    	\hline
    	$q=a\xi+a-m$ & $\xi$ is odd 	\\
    	\hline
    	                      & $q-2f(m\xi+m-1) \leq j\leq q-(\xi+1)-(2f-1)(m\xi+m-1)$,\\
    	      $m\equiv 1$ mod $4$  & $0\leq i \leq \xi $, $1\leq f \leq\lfloor \frac{m}{2} \rfloor$\\
    	\cline{2-2}
    	                     & $q-(m\xi +m-1) \leq j\leq q$, $0\leq i \leq \xi -1  $\\
    	                     
    	\hline
    	 & $1 \leq j\leq \lfloor \frac{m}{2} \rfloor\xi+\lfloor \frac{m}{2} \rfloor-1,$ \\	
    	& $\frac{m+1}{2}(\xi +1)+(2f-1)(m\xi+m-1)\leq j\leq \frac{m-1}{2}(\xi +1)-1+2f(m\xi +m-1),$ \\
    	 & $\frac{m+1}{2}(\xi+1)-\frac{m-3}{2}(m\xi+m-1) \leq j\leq \frac{a}{2}(\xi+1)-\frac{m+1}{2},$\\
    	$m\equiv 3$ mod $4$ & $q-2f(m\xi+m-1)-\frac{m-1}{2}(\xi +1)+1\leq j \leq q-(2f-1)(m\xi+m-1)-\frac{m+1}{2}(\xi+1),$\\
    	& $0\leq i \leq \xi $, $1\leq f \leq  \lfloor \frac{m}{4} \rfloor$\\
    	\cline{2-2}
    	& $q-\lfloor \frac{m}{2} \rfloor(\xi+1)+1\leq j\leq q$, $0\leq i\leq \xi-1 $\\
    	\hline
    	$q=a\xi+a-m$ & $\xi$ is even 	\\
    	\hline
    	& $1 \leq j\leq \lfloor \frac{m}{2} \rfloor\xi+\lfloor \frac{m}{2} \rfloor-1,$ \\	
    	& $\frac{m+1}{2}(\xi +1)+(2f-1)(m\xi+m-1)\leq j\leq \frac{m-1}{2}(\xi +1)-1+2f(m\xi +m-1),$ \\
    	& $ \frac{a}{2}\xi-\frac{m-1}{2}\leq j\leq (\frac{a}{2}+1)(\xi+1)-\frac{m-1}{2}-1,$\\
    	$m\equiv 1$ mod $4$ & $q-2f(m\xi+m-1)-\frac{m-1}{2}(\xi +1)+1\leq j \leq q-(2f-1)(m\xi+m-1)-\frac{m+1}{2}(\xi+1),$\\
    	& $0\leq i \leq \xi $, $1\leq f \leq  \lfloor \frac{m}{4} \rfloor$\\
    	\cline{2-2}
    	& $q-\lfloor \frac{m}{2} \rfloor(\xi+1)+1\leq j\leq q$, $0\leq i\leq \xi-1 $\\
    	\hline
    	& $q-2f(m\xi+m-1) \leq j\leq q-(\xi+1)-(2f-1)(m\xi+m-1)$,\\
    	$m\equiv 3$ mod $4$  & $0\leq i \leq \xi $, $1\leq f \leq\lfloor \frac{m}{2} \rfloor$\\
    	\cline{2-2}
    	& $q-(m\xi +m-1) \leq j\leq q$, $0\leq i \leq \xi -1  $\\
    	\hline
    \end{tabular}
\end{table}

   (2) $q=a\xi+m$, the range of $i$, $j$ is shown in Table 3.\\
    \begin{table}[h]
   	\setlength{\abovecaptionskip}{0.cm}
   	\setlength{\belowcaptionskip}{0.2cm}
   	\small
   	\centering
   	\caption{\newline   Parameters of $i$ and $j$.}
   	\begin{tabular}{c|c}
   		\hline
   		$q=a\xi+m$ & $\xi$ is odd 	\\
   		\hline
   		& $1 \leq j\leq m\xi $, $0\leq i \leq \xi $\\
   		\cline{2-2}
   		$m\equiv 1$ mod $4$  & $k+1+(2f-1)(m\xi+1)\leq j\leq mk+(2f-1)(m\xi+1),$\\
   		  & $1\leq f \leq  \lfloor \frac{m}{2} \rfloor$,$0\leq i\leq \xi-1 $\\
   		
   		\hline
   		& $0\leq i \leq \xi $, $1 \leq j\leq \lfloor \frac{m}{2} \rfloor\xi $ \\
   		\cline{2-2}
   		& $\frac{m+1}{2}\xi +1+(2f-1)(m\xi+1)\leq j\leq \frac{m-1}{2}\xi +2f(m\xi+1),$ \\	
   		& $ \frac{a}{2}\xi+\frac{m+1}{2}\leq j\leq (\frac{a}{2}+1)\xi+\frac{m-1}{2},$ \\
   		$m\equiv 3$ mod $4$ & $ q-2f(m\xi+1)-\frac{m-3}{2}\xi +1\leq j \leq q-(2f-1)(m\xi+1)-\frac{m+1}{2}\xi -1,$\\
   		 & $ q-\lfloor \frac{m}{2} \rfloor\xi\leq j\leq q.$\\
   		& $0\leq i \leq \xi-1 $, $1\leq f \leq  \lfloor \frac{m}{4} \rfloor$\\
   		\hline
   		$q=a\xi+m$ & $\xi$ is even 	\\
   		\hline
   		& $1 \leq j\leq \lfloor \frac{m}{2} \rfloor\xi $, $0\leq i\leq \xi $\\
   		\cline{2-2}
   		
   		& $\frac{m+1}{2}\xi +1+(2f-1)(m\xi+1)\leq j\leq \frac{m-1}{2}\xi +2f(m\xi+1),$ \\	
   		$m\equiv 1$ mod $4$ & $\frac{a}{2}\xi+\frac{m+1}{2}\leq j\leq (\frac{a}{2}+m-1)\xi+\frac{m+1}{2},$ \\
   		& $q-2f(m\xi+1)-\frac{m-3}{2}\xi +1\leq j \leq q-(2f-1)(m\xi+1)-\frac{m+1}{2}\xi -1,$\\
   		& $q-\lfloor \frac{m}{2} \rfloor\xi\leq j\leq q$\\
   		& $0\leq i \leq \xi-1 $, $1\leq f \leq  \lfloor \frac{m}{4} \rfloor$\\
   		\hline
   		& $1 \leq j\leq m\xi $, $0\leq i \leq \xi $\\
   		\cline{2-2}
   		$m\equiv 3$ mod $4$  & $k+1+(2f-1)(m\xi+1)\leq j\leq mk+(2f-1)(m\xi+1),$\\
   		& $1\leq f \leq  \lfloor \frac{m}{2} \rfloor$,$0\leq i\leq \xi-1 $\\
   		\hline
   	\end{tabular}
   \end{table} 

 Note that the subscript of $C_{s+2(iq+j)}$ are all belong to $[s,3s]$.
\begin{lemma}\label{le:4.2}
	~Let $n=\frac{2(q^2+1)}{a}$ and $s=\frac{n}{2}$, where $q$ is an odd prime power, $a=m^2+1$ and $m$ is an odd integer.\\
	(1) $q=a\xi+a-m$, \\
	If $m\equiv 1~{\rm mod}~4$ and $\xi$ is an odd (even) positive integer, or $m\equiv 3~{\rm mod}~4$ and $\xi$ is an even (odd) positive integer,
	then $$-qC_{s+2(iq+j)}=C_{{s+2(jq-i)}} (-qC_{s+2(iq+j)}=C_{{3s+2(jq-i)}}).$$ 
	(2) $q=a\xi+m$, \\
	If $m\equiv 1~{\rm mod}~4$ and $\xi$ is an odd (even) positive integer, or $m\equiv 3~{\rm mod}~4$ and $\xi$ is an even (odd) positive integer,
	then
	$$-qC_{s+2(iq+j)}=C_{{s+2(jq-i)}} (-qC_{s+2(iq+j)}=C_{{3s+2(jq-i)}}).$$ 
\end{lemma}
\begin{proof} 
 Since the proof of two cases is similar, then we only prove the case (1). Let $q=a\xi+a-m$, then we can obtain the following results.
	
	(a) Note that we have $C_{s+2(iq+j)}=\{s+2(iq+j), s-2(iq+j)\}$, for the corresponding range of $i$ and $j$ mentioned earlier. 
	
    To prove $-qC_{s+2(iq+j)}=C_{{s+2(jq-i)}}$, we first give the following calculus.
	\begin{equation*}
		\begin{split}
			-q\cdot(s-2(iq+j))&=-qs+2iq^2+2jq\\
			&=-(q+1-1)\cdot s+2i\cdot(q^2+1-1)+2jq\\
			&\equiv-\frac{q+1}{4}\cdot2n+{s+2(jq-i)}~{\rm mod}~2n.
		\end{split}
	\end{equation*}	
    If $m\equiv 1$ mod $4$ and $\xi$ is an odd positive integer, then let $m=4h+1$. Therefore,
    \begin{equation*}
    	\begin{split}
    		q+1&=a\xi+a-m+1\\
    		&\equiv 2(\xi+1)~{\rm mod}~4\\
    		&\equiv 0~{\rm mod}~4.\\
    	\end{split}
    \end{equation*}	\\
    If $m\equiv 3$ mod $4$ and $\xi$ is an even positive integer, then let $m=4h+3$. Therefore,
    \begin{equation*}
    	\begin{split}
    		q+1&=a\xi+a-m+1\\
    		&\equiv 10\xi+8~{\rm mod}~4\\
    		&\equiv 0~{\rm mod}~4.\\
    	\end{split}
    \end{equation*}	
	The desired result $-qC_{s+2(iq+j)}=C_{{s+2(jq-i)}}$ follows.\\
	(b) Similarly, we have $C_{s+2(iq+j)}=\{s+2(iq+j), s-2(iq+j)\}$, for the corresponding range of $i$ and $j$ mentioned earlier. 
	To prove $-qC_{s+2(iq+j)}=C_{{3s+2(jq-i)}}$, we first give the following calculus.
	\begin{equation*}
		\begin{split}
			-q\cdot(s-2(iq+j))&=-qs+2iq^2+2jq\\
			&=-(q+1-1)\cdot s+2i\cdot(q^2+1-1)+2qj\\
			&\equiv{-\frac{q-1}{4}\cdot 2n+3s+2(jq-i)}~{\rm mod}~2n.\\
		\end{split}
	\end{equation*}	
	If $m\equiv 1$ mod $4$ and $\xi$ is an even positive integer, then let $m=4h+1$. Therefore,
	\begin{equation*}
	\begin{split}
	q-1&=a\xi+a-m+1\\
	&\equiv 2\xi~{\rm mod}~4\\
	&\equiv 0~{\rm mod}~4.\\
	\end{split}
	\end{equation*}	\\
	If $m\equiv 3$ mod $4$ and $\xi$ is an odd positive integer, then let $m=4h+3$. Therefore,
	\begin{equation*}
	\begin{split}
	q-1&=a\xi+a-m+1\\
	&\equiv 10\xi+6~{\rm mod}~4\\
	&\equiv 0~{\rm mod}~4.\\
	\end{split}
	\end{equation*}	
	The desired result $-qC_{s+2(iq+j)}=C_{{3s+2(jq-i)}}$ follows.
	
\end{proof}
\



\noindent\textbf{Case \uppercase\expandafter{\romannumeral 1} \ \ $q=a\xi+a-m$}\\

As we all know, one needs to calculate the number of entangled states $c$ when they try to construct EAQECCs. Thus, we give the following preparation lemma to determine its parameters.

\begin{lemma}\label{le:4.3}
 	Let $n=\frac{2(q^2+1)}{a}$ and $s=\frac{n}{2}$, where $q=a\xi+a-m$ is an odd prime power, $a=m^2+1$ and $m$ is an odd integer. \\
    (1) If $m\equiv 1~{\rm mod}~4$ and $\xi$ is an odd positive integer, or $m\equiv 3~{\rm mod}~4$ and $\xi$ is an even positive integer. For a positive integer $\alpha$ with $1\leq \alpha\leq \xi $, let
    \iffalse
 	\begin{equation*}
 	\begin{split}
 	T_{1}~~&=\bigcup_{\substack{0 \leq l \leq m-1}}\bigcup_{\substack{2t_l\xi+h_l+ \alpha\leq j\leq 2(t_l+1)\xi+h_l-1 -\alpha,\\a_{1l}\lceil \frac{m}{2} \rceil+b_{1l}\lfloor \frac{m}{2} \rfloor\leq t_l \leq a_{2l}\lceil \frac{m}{2} \rceil+b_{2l}\lfloor \frac{m}{2} \rfloor,~h_l=2t_l-l+1,\\if~  j\leq (a-m)\xi+(a-2m),~0 \leq i\leq\alpha,~else,~0 \leq i\leq\alpha-1}}C_{s+2(iq+j)}\\
 	\end{split}
 	\end{equation*}
    \fi
    \begin{equation*}
    	\begin{split}
    		T_{1}~~&=\bigcup_{\substack{2t_l\xi+h_l+ \alpha\leq j\leq 2(t_l+1)\xi+h_l-1 -\alpha,\\if~  j\leq (a-m)\xi+(a-2m),~0 \leq i\leq\alpha,~else,~0 \leq i\leq\alpha-1}}C_{s+2(iq+j)}\\
    	\end{split}
    \end{equation*}
	Then $T_1\cap-qT_1=\emptyset.$ where $0 \leq l \leq m-1$, $a_{1l}\lceil \frac{m}{2} \rceil+b_{1l}\lfloor \frac{m}{2} \rfloor\leq t_l \leq a_{2l}\lceil \frac{m}{2} \rceil+b_{2l}\lfloor \frac{m}{2} \rfloor,~h_l=2t_l-l+1$, $a_{10}=a_{20}=b_{10}=0$, $b_{20}=1$, and  $a_{1l}=a_{1l-1}+l ~{\rm mod} ~2 ,b_{1l}=b_{1l-1}+(l+1) ~{\rm mod} ~2,a_{2l}=a_{2l-1}+(l+1) ~{\rm mod} ~2,b_{2l}=b_{2l-1}+l ~{\rm mod} ~2$.
 	\noindent(2) If $m\equiv 1~{\rm mod}~4$ and $\xi$ is an even positive integer, or $m\equiv 3~{\rm mod}~4$ and $\xi$ is an odd positive integer. Let \\
 	\begin{equation*}
 		\begin{split}
 		T_{1}~~= ~~~&\bigcup_{\substack{1 \leq j \leq \xi-\alpha, \\0 \leq i\leq\alpha}}C_{s+2(iq+j)}
 		\cup\bigcup_{\substack{m^2\xi +m^2-m+2+\alpha \leq j \leq      q,\\0 \leq i\leq\alpha-1}}C_{s+2(iq+j)}\\
 		&\cup\bigcup_{\substack{0 \leq l \leq m-1}}~~\cup\bigcup_{\substack{(2t_l-1)\xi+h_l+ \alpha\leq j\leq (2t_l+1)\xi+h_l-1 -\alpha,\\if~  j\leq (a-m)\xi+(a-2m),~0 \leq i\leq\alpha,~else,~0 \leq i\leq\alpha-1}}C_{s+2(iq+j)}\\
 		\end{split}
 	\end{equation*}
 	Then $T_1\cap-qT_1=\emptyset.$ where $0 \leq l \leq m-1$, $a_{1l}\lceil \frac{m}{2} \rceil+b_{1l}\lfloor \frac{m}{2} \rfloor\leq t_l \leq a_{2l}\lceil \frac{m}{2} \rceil+b_{2l}\lfloor \frac{m}{2} \rfloor,~h_l=2t_l-l$, $a_{10}=a_{20}=b_{10}=0$, $b_{20}=1$, and  $a_{1l}=a_{1l-1} ,b_{1l}=b_{1l-1},a_{2l}=a_{2l-1}+l ~{\rm mod} ~2 ,b_{2l}=b_{2l-1}+(l+1) ~{\rm mod} ~2$.

\end{lemma}

The proof of Lemma 4.3 is very complicated and is placed in Appendix.  \\

\noindent\textbf{Example 1} Let $m=3$, then $a=10$, $q=10\xi+7$, $n=\frac{q^2+1}{5}$ and $s=\frac{n}{2}$. For a positive integer $\alpha$ with $1\leq \alpha\leq \xi $. Let $\xi$ is an even positive integer. According Lemma 4.3 (1), we have $0 \leq l \leq 2$. Below we use the following Table $1$ to give the range of each parameter to get $T_1$. We have
    \begin{equation*}
    \begin{split}
    T_{1}~=&~\bigcup_{\substack{1+\alpha \leq j \leq 2\xi-\alpha, \\0 \leq i\leq\alpha}}C_{s+2(iq+j)}
    \cup\bigcup_{\substack{2\xi+3+\alpha \leq j \leq 4\xi+2-\alpha, \\0 \leq i\leq\alpha}}C_{s+2(iq+j)}
    \cup\bigcup_{\substack{4\xi+4+\alpha \leq j \leq 6\xi+3-\alpha, \\0 \leq i\leq\alpha}}C_{s+2(iq+j)}\\
    &\cup\bigcup_{\substack{6\xi+5+\alpha \leq j \leq 7\xi+4, \\0 \leq i\leq\alpha}}C_{s+2(iq+j)}
    \cup\bigcup_{\substack{7\xi+5 \leq j \leq 8\xi+4-\alpha, \\0 \leq i\leq\alpha-1}}C_{s+2(iq+j)}
    ~~\cup\bigcup_{\substack{8\xi+7 \leq j \leq 10\xi+6-\alpha, \\0 \leq i\leq\alpha-1}}C_{s+2(iq+j)}\\
    \end{split}
    \end{equation*}
  \\
 Then it is easy to check that $T_1\cap-qT_1=\emptyset$.
 
		 
		

\begin{theorem} \label{th:4.4}
	Let $n=\frac{2(q^2+1)}{a}$ and $s=\frac{n}{2}$, where $q=a\xi+a-m$ is an odd prime power, $a=m^2+1$ and $m$ is an odd integer. For a positive integer $\alpha$ with $1\leq \alpha\leq \xi $, let $\mathcal{C}$ be a negacyclic code with defining set $Z$ given as follows $$Z=C_s\cup C_{s+2}\cup\dots\cup C_{s+2[\alpha q+(a-m)\xi +a-2m]}.$$
	If $m\equiv 1~{\rm mod}~4$ and $\xi$ is an odd (even) positive integer, or $m\equiv 3~{\rm mod}~4$ and $\xi$ is an even (odd) positive integer, 
	then $|Z_{1}|=2\alpha[(a\alpha +2(a-m)]+2a-4m+1$ ($|Z_{1}|=2\alpha[(a\alpha +2(a-m)]+2a-4m$).\\
\end{theorem}

\begin{proof}
	For brevity, let's just show the first case, and the second case is similar. Let  \\
	\begin{equation*}
		\begin{split}
			T_{1}~~&=\bigcup_{\substack{2t_l\xi+h_l+ \alpha\leq j\leq 2(t_l+1)\xi+h_l-1 -\alpha,\\if~  j\leq (a-m)\xi+(a-2m),~0 \leq i\leq\alpha,~else,~0 \leq i\leq\alpha-1}}C_{s+2(iq+j)}\\
		\end{split}
	\end{equation*}
	where $0 \leq l \leq m-1$, $a_{1l}\lceil \frac{m}{2} \rceil+b_{1l}\lfloor \frac{m}{2} \rfloor\leq t_l \leq a_{2l}\lceil \frac{m}{2} \rceil+b_{2l}\lfloor \frac{m}{2} \rfloor,~h_l=2t_l-l+1$, $a_{10}=a_{20}=b_{10}=0$, $b_{20}=1$, and  $a_{1l}=a_{1l-1}+l ~{\rm mod} ~2 ,b_{1l}=b_{1l-1}+(l+1) ~{\rm mod} ~2,a_{2l}=a_{2l-1}+(l+1) ~{\rm mod} ~2,b_{2l}=b_{2l-1}+l ~{\rm mod} ~2$.
	and
	\begin{equation*}
		\begin{split}
			T'_{1}~~=&C_{s}
			~~~\bigcup_{\substack{1 \leq j \leq \alpha,\\0\leq i\leq \xi }}C_{s+2(iq+j)}
			~~~~\cup\bigcup_{\substack{a\xi+a-m-\alpha \leq j \leq q,\\0\leq i\leq \xi-1 }}C_{s+2(iq+j)}\\
			&~~~~~\cup\bigcup_{\substack{0 \leq l \leq m-1}}
		\cup\bigcup_{\substack{2t_l(\xi+1)-h_l- \alpha\leq j\leq (2t_l+1)(\xi+1)-h_l'-2 +\alpha,\\a_{1l}\lceil \frac{m}{2} \rceil+b_{1l}\lfloor \frac{m}{2} \rfloor < t_l \leq a_{2l}\lceil \frac{m}{2} \rceil+b_{2l}\lfloor \frac{m}{2} \rfloor,~h_l=2t_l-l+1,\\if~  j\leq (a-m)\xi+(a-2m),~0 \leq i\leq\alpha,~else,~0 \leq i\leq\alpha-1}}C_{s+2(iq+j)}\\
			&~~~~
			\cup\bigcup_{\substack{1 \leq l' \leq m-1}}
			~\cup\bigcup_{\substack{2t_l'\xi+h_l'- \alpha\leq j\leq (2t_l'+1)\xi+h_l'-2 +\alpha,\\t_l'=\frac{l'+\delta }{2},~h_l'=2t_l'-l',\\if~  j\leq (a-m)\xi+(a-2m),~0 \leq i\leq\alpha,~else,~0 \leq i\leq\alpha-1}}C_{s+2(iq+j)}\\
		\end{split}
	\end{equation*}
	where $a_{10}=a_{20}=b_{10}=0$, $b_{20}=1$, and when $l$ is odd, we have \\
	$$\begin{cases}a_{1l}=a_{1l-1}+1,\\b_{1l}=b_{1l-1},\\a_{2l}=a_{2l-1},\\b_{2l}=b_{2l-1}+1,\end{cases}$$
	when  $l$ is even, we have
	$$\begin{cases}a_{1l}=a_{1l-1},\\b_{1l}=b_{1l-1}+1,\\a_{2l}=a_{2l-1}+1,\\b_{2l}=b_{2l-1},\end{cases}$$
	and $$\delta=\begin{cases}0,~~l'~is ~odd,\\1,~~l'~ is ~even.\end{cases}$$\\
	From Lemma \ref{le:4.2}, we have 
\begin{equation*}
	\begin{split}
		-qT'_{1}~~=&C_{s}
		~~~\cup\bigcup_{\substack{1 \leq j \leq \alpha,\\0\leq i\leq \xi }}C_{s+2(jq-i)}
		~~~~\cup\bigcup_{\substack{a\xi+a-m-\alpha \leq j \leq q,\\0\leq i\leq \xi-1 }}C_{s+2(jq-i)}\\
		&~~~~~\cup\bigcup_{\substack{0 \leq l \leq m-1}}
		\cup\bigcup_{\substack{2t_l(\xi+1)-h_l- \alpha\leq j\leq (2t_l+1)(\xi+1)-h_l'-2 +\alpha,\\a_{1l}\lceil \frac{m}{2} \rceil+b_{1l}\lfloor \frac{m}{2} \rfloor < t_l \leq a_{2l}\lceil \frac{m}{2} \rceil+b_{2l}\lfloor \frac{m}{2} \rfloor,~h_l=2t_l-l+1,\\if~  j\leq (a-m)\xi+(a-2m),~0 \leq i\leq\alpha,~else,~0 \leq i\leq\alpha-1}}C_{s+2(jq-i)}\\
		&~~~~
		\cup\bigcup_{\substack{1 \leq l' \leq m-1}}
		~\cup\bigcup_{\substack{2t_l'\xi+h_l'- \alpha\leq j\leq (2t_l'+1)\xi+h_l'-2 +\alpha,\\t_l'=\frac{l'+\delta }{2},~h_l'=2t_l'-l',\\if~  j\leq (a-m)\xi+(a-2m),~0 \leq i\leq\alpha,~else,~0 \leq i\leq\alpha-1}}C_{s+2(jq-i)}\\
	\end{split}
\end{equation*}
	
	It is easy to check that $-qT_1'=T_1'$. From the definitions of $Z$, $T_1$ and $T_1'$, we have $Z=T_1\cup T_1'$. 
	Then from the definition of $Z_{1}$,
	\begin{equation*}
	\begin{split}
	Z_{1}=Z\cap (-qZ)&=(T_1\cup T_1')\cap(-qT_1\cup -qT_1')\\
	&=(T_1\cap-qT_1)\cup(T_1\cap-qT_1')\cup(T_1'\cap-qT_1)\cup(T_1'\cap-qT_1')\\
	&=T_1'.
	\end{split}
	\end{equation*}
	Therefore, $|Z_{1}|=|T_1'|=2\alpha(a\alpha +2a-2m)+2a-4m+1$.

\end{proof}
\ 

From Lemmas \ref{le:4.2}, \ref{le:4.3} and Theorem \ref{th:4.4} above, we can obtain the first construction of EAQECCs in the following theorem.


\begin{theorem} \label{th:4.5}
	Let $n=\frac{2(q^2+1)}{a}$ and $s=\frac{n}{2}$, where $q=a\xi+a-m$ is an odd prime power, $a=m^2+1$ and $m$ is an odd integer.
	If $m\equiv 1~{\rm mod}~4$ and $\xi$ is an odd (even) positive integer, or $m\equiv 3~{\rm mod}~4$ and $\xi$ is an even (odd) positive integer, then there are EAQECCs with parameters $$[[n, k, d; c]]_q([[n, k-1, d; c-1]]_q),$$ 
	where 
	$$k=n-4\alpha q+4(a-m)(\alpha -\xi ) +2a\alpha^2-2a+4m-1,$$
	$$d=2[\alpha q+(a-m)\xi+a-2m +1],$$ 
	$$c=2\alpha[(a\alpha +2(a-m)]+2a-4m+1,$$ 
	and $\alpha$ is a positive integer satisfying $1\leq \alpha\leq \xi $. Besides, they are EAQMDS codes if $d\leq\frac{n+2}{2}$.\\
\end{theorem}

\begin{proof}
	For brevity, let's just show the first case, and the second case is similar. For a positive integer $\alpha$ with $1\leq \alpha\leq \xi$, suppose   
	that $\mathcal{C}$ is a negacyclic code of length $n$ with defining set        
	$$Z=C_s\cup C_{s+2}\cup\dots\cup C_{s+2[\alpha q+(a-m)\xi +a-2m]},$$
	where $\xi$, $a$, $m$, $q$ and $s$  are defined as above.
	
	Then the dimension of $\mathcal{C}$ is $n-2[\alpha q+(a+m)\xi ]+4m-2a-1$.
	Observe that negacyclic code $\mathcal{C}$ have $2[\alpha q+(a-m)\xi+a-2m]+1$ consecutive roots.
	By Theorem \ref{th:2.2}, the minimum distance of $\mathcal{C}$ is at least $2[\alpha q+(a-m)\xi+a-2m +1]$. Then from Theorem \ref{th:2.1}, $\mathcal{C}$ is an MDS code with parameters $[n,n-2[\alpha q+(a+m)\xi ]+4m-2a-1, 2[\alpha q+(a-m)\xi+a-2m +1]~]$.
	Furthermore, we have $|Z_{1}|=|T_1'|=2\alpha(a\alpha +2a-2m)+2a-4m+1$ from Theorem \ref{th:4.4}.
	By Theorem \ref{th:3.1}, there are EAQECCs with parameters $[[n, k, d; c]]_q$, 	where 
	$$k=n-4\alpha q+4(a-m)(\alpha -\xi ) +2a\alpha^2-2a+4m-1,$$
	$$d=2[\alpha q+(a-m)\xi+a-2m +1],$$ 
	$$c=2\alpha[(a\alpha +2(a-m)]+2a-4m+1.$$
	It is easy to check that 
	$$n+c-k=4[\alpha q+(a-m)\xi+a-2m +1]-2=2(d_1-1).$$
	Therefore, it implies that the EAQECCs are EAQMDS codes if $d\leq\frac{n+2}{2}$ by Theorem \ref{th:3.2}.
	
\end{proof}
\

\noindent\textbf{Example 2}
~Let $m=1$,~then~$a=m^2+1=2$ and~$q=a\xi +a-m=2\xi +1$, then $1\leq \alpha\leq \xi =5,$ we have $q=11$.
Then $n=q^2+1=122$, and $~s=\frac{q^2+1}{2}=61$. 
According to Lemma \ref{le:4.3}, we have $l=0$, $t_0=0$, and $h_0=1$, we can obtain
\begin{equation*}
	\begin{split}
		T_{1}=\bigcup_{\substack{1+\alpha\leq j \leq 9-\alpha,\\if~ j\leq 5,~0\leq u\leq \alpha ,\\else,~0\leq u\leq \alpha-1}}C_{s+2(iq+j)},
	\end{split}
\end{equation*}
where $\alpha=1, 2,3,4,5$ respectively. It is easy to check that $T_1\cap-qT_1=\emptyset$.
Then according to Theorem \ref{th:4.5}, we can obtain following EAQECCs: $$[[122,65,34;9]]_{11}, [[122,37,56;25]]_{11}, [[122,17,78;49]]_{11},$$  $$[[122,5,100;81]]_{11}, [[122,1,122;121]]_{11}.$$ Especially, they are EAQMDS codes when $d\leq 62$.\\

\noindent\textbf{Example 3}
~Let $m=3$,~then~$a=m^2+1=10$~and~$q=a\xi +a-m=10\xi +7.$ Let $1\leq \alpha\leq \xi =3,$ we have $q=37$.
Then $n=\frac{q^2+1}{5}=274$, and $~s=\frac{q^2+1}{10}=137$.
According to Lemma \ref{le:4.3}, we have $0 \leq l \leq 2$ and 
$$\begin{cases} 0<t_0\leq 1,~h_0=2t_0\\1<t_1\leq 3,~h_1=2t_1-1\\3<t_2\leq 4,~h_2=2t_2-2\\\end{cases}$$ we can obtain

\begin{equation*}
	\begin{split}
		T_{1}~~=&\bigcup_{\substack{1\leq j \leq 3-\alpha,\\~0\leq i\leq \alpha}}C_{s+2(iq+j)}
		\cup\bigcup_{\substack{3(2t_0-1)+2t_0+ \alpha\leq j\leq 3(2t_0+1)+2t_0-1 -\alpha,\\ if~  j\leq 25,~0 \leq i\leq\alpha,~else,~0 \leq i\leq\alpha-1}}C_{s+2(iq+j)}\\
		&\cup\bigcup_{\substack{3(2t_1-1)+2t_1-1+ \alpha\leq j\leq 3(2t_1+1)+2t_1-2 -\alpha,\\ if~  j\leq 25,~0 \leq i\leq\alpha,~else,~0 \leq i\leq\alpha-1}}C_{s+2(iq+j)}\\
		&\cup\bigcup_{\substack{3(2t_2-1)+2t_2-2+ \alpha\leq j\leq 3(2t_2+1)+2t_2-3 -\alpha,\\ if~  j\leq 25,~0 \leq i\leq\alpha,~else,~0 \leq i\leq\alpha-1}}C_{s+2(iq+j)}
		\cup\bigcup_{\substack{35+\alpha \leq j \leq q,\\~0\leq i\leq \alpha-1}}C_{s+2(iq+j)}
	\end{split}
\end{equation*}
where $\alpha=1,2,3$ respectively. It is easy to check that $T_1\cap-qT_1=\emptyset$.
   Then according to Theorem \ref{th:4.5}, we can obtain EAQECCs as follows: $$[[274,80,126;56]]_{37}, [[274,20,200;144]]_{37}, [[274,0,274;272]]_{37}.$$
   Especially, they are EAQMDS codes when $d\leq 138$.\\

\noindent\textbf{Case \uppercase\expandafter{\romannumeral 2} $~~~~q=a\xi+m$}\\

As for the case that $n=\frac{2(q^2+1)}{a}$ and $q=a\xi+m$, where $a=m^2+1$ and $m$ is an odd integer, we can produce the following EAQECCs. The proof is similar to that in the Case \uppercase\expandafter{\romannumeral 1}, so we omit it here for brevity.

\begin{lemma}\label{le:4.6}
 	Let $n=\frac{2(q^2+1)}{a}$ and $s=\frac{n}{2}$, where $q=a\xi+m$ is an odd prime power, $a=m^2+1$ and $m$ is an odd integer. For a positive integer $\alpha$ with $1\leq \alpha\leq \xi $.\\
 	(1) If $m\equiv 1~{\rm mod}~4$ and $\xi$ is an odd positive integer, or $m\equiv 3~{\rm mod}~4$ and $\xi$ is an even positive integer. Let \\
 	\iffalse
 	\begin{equation*}
 		\begin{split}
 			T_{1}~~&=
 			\bigcup_{\substack{0 \leq l \leq m-1}}\cup\bigcup_{\substack{2t_l\xi+h_l+ \alpha\leq j\leq 2(t_l+1)\xi+h_l-1 -\alpha,\\a_{1l}\lceil \frac{m}{2} \rceil+b_{1l}\lfloor \frac{m}{2} \rfloor\leq t_l \leq a_{2l}\lceil \frac{m}{2} \rceil+b_{2l}\lfloor \frac{m}{2} \rfloor,~h_l=l+1,\\if~  j\leq m\xi,~0 \leq i\leq\alpha,~else,~0 \leq i\leq\alpha-1}}C_{s+2(iq+j)}\\
 		\end{split}
 	\end{equation*}
 	where $a_{10}=a_{20}=b_{10}=0$, $b_{20}=1$, and when $l$ is odd, we have\\
 	$$\begin{cases}a_{1l}=a_{1l-1}+1,\\b_{1l}=b_{1l-1},\\a_{2l}=a_{2l-1},\\b_{2l}=b_{2l-1}+1,\end{cases}$$ 
 	when $l$ is even, we have
 	$$\begin{cases}a_{1l}=a_{1l-1},\\b_{1l}=b_{1l-1}+1,\\a_{2l}=a_{2l-1}+1,\\b_{2l}=b_{2l-1}.\end{cases}$$
 	Then $T_1\cap-qT_1=\emptyset.$\\
 	\fi
 	\begin{equation*}
 		\begin{split}
 			T_{1}~~&=\bigcup_{\substack{2t_l\xi+h_l+ \alpha\leq j\leq 2(t_l+1)\xi+h_l-1 -\alpha,\\if~  j\leq m\xi,~0 \leq i\leq\alpha,~else,~0 \leq i\leq\alpha-1}}C_{s+2(iq+j)}\\
 		\end{split}
 	\end{equation*}
 	Then $T_1\cap-qT_1=\emptyset.$ where $0 \leq l \leq m-1$, $a_{1l}\lceil \frac{m}{2} \rceil+b_{1l}\lfloor \frac{m}{2} \rfloor\leq t_l \leq a_{2l}\lceil \frac{m}{2} \rceil+b_{2l}\lfloor \frac{m}{2} \rfloor,~h_l=l+1$, $a_{10}=a_{20}=b_{10}=0$, $b_{20}=1$, and  $a_{1l}=a_{1l-1}+l ~{\rm mod} ~2 ,b_{1l}=b_{1l-1}+(l+1) ~{\rm mod} ~2,a_{2l}=a_{2l-1}+(l+1) ~{\rm mod} ~2,b_{2l}=b_{2l-1}+l ~{\rm mod} ~2$.
 	
 	(2) If $m\equiv 1~{\rm mod}~4$ and $\xi$ is an even positive integer, or $m\equiv 3~{\rm mod}~4$ and $\xi$ is an odd positive integer. Let \\
 	\iffalse
 	\begin{equation*}
 		\begin{split}
 			T_{1}~~&=~\bigcup_{\substack{1 \leq j \leq \xi-\alpha, \\0 \leq i\leq\alpha}}C_{s+2(iq+j)}
 			\cup\bigcup_{\substack{0 \leq l \leq m-1}}~~\cup\bigcup_{\substack{(2t_l-1)\xi+h_l+ \alpha\leq j\leq (2t_l+1)\xi+h_l-1 -\alpha,\\a_{1l}\lceil \frac{m}{2} \rceil+b_{1l}\lfloor \frac{m}{2} \rfloor < t_l \leq a_{2l}\lceil \frac{m}{2} \rceil+b_{2l}\lfloor \frac{m}{2} \rfloor,~h_l=l+1,\\if~  j\leq m\xi,~0 \leq i\leq\alpha,~else,~0 \leq i\leq\alpha -1}}C_{s+2(iq+j)}\\
 			&\cup\bigcup_{\substack{m^2\xi +m+1+\alpha \leq j \leq      q,\\0 \leq i\leq\alpha-1}}C_{s+2(iq+j)}\\
 		\end{split}
 	\end{equation*}
 	where $a_{10}=a_{20}=b_{10}=0$, $b_{20}=1$, and when  $l$ is odd, we have \\
  	$$\begin{cases}a_{1l}=a_{2l-1},\\b_{1l}=b_{2l-1},\\a_{2l}=a_{1l}+1,\\b_{2l}=b_{1l},\end{cases}$$
  	when $l$ is even, we have
  	$$\begin{cases}a_{1l}=a_{2l-1},\\b_{1l}=b_{2l-1},\\a_{2l}=a_{1l},\\b_{2l}=b_{1l}+1.\end{cases}$$
 	Then $T_1\cap-qT_1=\emptyset.$
 	\fi
 	\begin{equation*}
 		\begin{split}
 			T_{1}~~= ~~~&\bigcup_{\substack{1 \leq j \leq \xi-\alpha, \\0 \leq i\leq\alpha}}C_{s+2(iq+j)}
 			\cup\bigcup_{\substack{m^2\xi +m+1+\alpha \leq j \leq q,\\0 \leq i\leq\alpha-1}}C_{s+2(iq+j)}\\
 			&\cup\bigcup_{\substack{0 \leq l \leq m-1}}~~\cup\bigcup_{\substack{(2t_l-1)\xi+h_l+ \alpha\leq j\leq (2t_l+1)\xi+h_l-1 -\alpha,\\if~  j\leq m\xi,~0 \leq i\leq\alpha,~else,~0 \leq i\leq\alpha-1}}C_{s+2(iq+j)}\\
 		\end{split}
 	\end{equation*}
 	Then $T_1\cap-qT_1=\emptyset.$ where $0 \leq l \leq m-1$, $a_{1l}\lceil \frac{m}{2} \rceil+b_{1l}\lfloor \frac{m}{2} \rfloor\leq t_l \leq a_{2l}\lceil \frac{m}{2} \rceil+b_{2l}\lfloor \frac{m}{2} \rfloor,~h_l=l+l$, $a_{10}=a_{20}=b_{10}=0$, $b_{20}=1$, and  $a_{1l}=a_{1l-1} ,b_{1l}=b_{1l-1},a_{2l}=a_{2l-1}+l ~{\rm mod} ~2 ,b_{2l}=b_{2l-1}+(l+1) ~{\rm mod} ~2$.
 	\end{lemma}

 \noindent\textbf{Example 4} Let $m=5$, then $a=26$, $q=26\xi+5$, $n=\frac{q^2+1}{13}$ and $s=\frac{n}{2}$. For a positive integer $\alpha$ with $1\leq \alpha\leq \xi $. Let $\xi$ is an even positive integer. According to (2) of Lemma 4.6 , we have $0 \leq l \leq 4$. Below we use the following Table $2$ to give the range of each parameter to get $T_1$. We have\\
 \begin{equation*}
 \begin{split}
T_{1}~~=&~~~~~\bigcup_{\substack{1 \leq j \leq \xi-\alpha, \\0 \leq i\leq\alpha}}C_{s+2(iq+j)}
~~~~~~
\cup\bigcup_{\substack{\xi+1+\alpha \leq j \leq 3\xi-\alpha, \\0 \leq i\leq\alpha}}C_{s+2(iq+j)}
~~
\cup\bigcup_{\substack{3\xi+1+\alpha \leq j \leq 5\xi-\alpha, \\0 \leq i\leq\alpha}}C_{s+2(iq+j)}\\
&\cup\bigcup_{\substack{5\xi+2+\alpha \leq j \leq 7\xi+1-\alpha, \\0 \leq i\leq\alpha-1}}C_{s+2(iq+j)}
\cup\bigcup_{\substack{7\xi+2+\alpha \leq j \leq 9\xi+1, \\0 \leq i\leq\alpha-1}}C_{s+2(iq+j)}
~~
\cup\bigcup_{\substack{9\xi+2 \leq j \leq 11\xi+1-\alpha, \\0 \leq i\leq\alpha-1}}C_{s+2(iq+j)}\\
&\cup\bigcup_{\substack{11\xi+3 \leq j \leq 13\xi+2-\alpha, \\0 \leq i\leq\alpha-1}}C_{s+2(iq+j)}
\cup\bigcup_{\substack{13\xi+3 \leq j \leq 15\xi+2-\alpha, \\0 \leq i\leq\alpha-1}}C_{s+2(iq+j)}
\cup\bigcup_{\substack{15\xi+4 \leq j \leq 17\xi+3-\alpha, \\0 \leq i\leq\alpha-1}}C_{s+2(iq+j)}\\
&\cup\bigcup_{\substack{17\xi+4 \leq j \leq 19\xi+3-\alpha, \\0 \leq i\leq\alpha-1}}C_{s+2(iq+j)}
\cup\bigcup_{\substack{19\xi+4 \leq j \leq 21\xi+3-\alpha, \\0 \leq i\leq\alpha-1}}C_{s+2(iq+j)}
\cup\bigcup_{\substack{21\xi+5 \leq j \leq 23\xi+4-\alpha, \\0 \leq i\leq\alpha-1}}C_{s+2(iq+j)}\\
&\cup\bigcup_{\substack{23\xi+5 \leq j \leq 25\xi+4-\alpha, \\0 \leq i\leq\alpha-1}}C_{s+2(iq+j)}
~
\cup\bigcup_{\substack{25\xi+6 \leq j \leq 26\xi+5, \\0 \leq i\leq\alpha-1}}C_{s+2(iq+j)}
 \end{split}
 \end{equation*}
 \\
 Then it is easy to check that $T_1\cap-qT_1=\emptyset$.

\begin{theorem} \label{th:4.7}
	Let $n=\frac{2(q^2+1)}{a}$ and $s=\frac{n}{2}$, where $q=a\xi+m$ is an odd prime power, $a=m^2+1$ and $m$ is an odd integer. For a positive integer $\alpha$ with $1\leq \alpha\leq \xi $, let $\mathcal{C}$ be a negacyclic code with defining set $Z$ given as follows $$Z=C_s\cup C_{s+2}\cup\dots\cup C_{s+2[\alpha q+m\xi]}.$$
	If $m\equiv 1~{\rm mod}~4$ and $\xi$ is an odd (even) positive integer, or $m\equiv 3~{\rm mod}~4$ and $\xi$ is an even (odd) positive integer, 	
	then $|Z_{1}|=2\alpha(a\alpha +2m)+1$ ($|Z_{1}|=2\alpha(a\alpha +2m)$).\\
\end{theorem}

\begin{theorem} \label{th:4.8}
	Let $n=\frac{2(q^2+1)}{a}$ and $s=\frac{n}{2}$, where $q=a\xi+m$ is an odd prime power, $a=m^2+1$ and $m$ is an odd integer.
	If $m\equiv 1~{\rm mod}~4$ and $\xi$ is an odd (even) positive integer, or $m\equiv 3~{\rm mod}~4$ and $\xi$ is an even (odd) positive integer, 
	then there are EAQECCs with parameters $$[[n, k, d; c]]_q([[n, k-1, d; c-1]]_q),$$ 
	where 
	$$k=n-4\alpha q+4m(\alpha-\xi)+2a\alpha^2-1,$$
	$$d=2(\alpha q+m\xi +1),$$ 
	$$c=2\alpha(a\alpha +2m)+1,$$ 
	and $\alpha$ is a positive integer satisfying $1\leq \alpha\leq \xi $. Besides, they are EAQMDS codes if $d\leq\frac{n+2}{2}$.\\
\end{theorem}

\noindent\textbf{Example 5}
~Let $m=7$,~then~$a=m^2+1=50$~and~$q=a\xi +m=50\xi +7.$ Let $1\leq \alpha\leq \xi =2,$ we have $q=107$.
Then $n=\frac{q^2+1}{25}=458$ and $~s=\frac{q^2+1}{50}=229$. 
According to Lemma \ref{le:4.6}, we have $0 \leq l \leq 6$ and 
$$\begin{cases} 0\leq t_0\leq 3,~h_0=1\\4\leq t_1\leq 6,~h_1=2\\7\leq t_2\leq 10,~h_2=3\\11\leq t_3\leq 13,~h_3=4\\14\leq t_4\leq 17,~h_4=5\\18\leq t_5\leq 20,~h_5=6\\21\leq t_6\leq 24,~h_6=7.\end{cases}$$ We can obtain

\begin{equation*}
	\begin{split}
		T_{1}=	\bigcup_{\substack{0 \leq l \leq 6}}~~
		\cup\bigcup_{\substack{4t_l+h_l+ \alpha\leq j\leq 4(t_l+1)+h_l-1 -\alpha,\\if~  j\leq 14,~0 \leq i\leq\alpha,~else,~0 \leq i\leq\alpha-1}}C_{s+2(iq+j)}\\
	\end{split}
\end{equation*}
where $\alpha=1, 2$ respectively. It is easy to check that $T_1\cap-qT_1=\emptyset$.
Then according to Theorem \ref{th:4.8}, we can obtain EAQECCs as follows: $$[[458,101,244;129]]_{107}, [[458,1,458;457]]_{107}.$$ Especially, they are EAQMDS codes when $d\leq 230$.\\

\noindent\textbf{Example 6}
~Let $m=5$,~then~$a=m^2+1=26$~and~$q=a\xi +m=26\xi +5.$ Let $1\leq \alpha\leq \xi =4,$ we have $q=109$.
Then $n=\frac{q^2+1}{13}=914$ and $~s=\frac{q^2+1}{50}=457$. 
According to Lemma \ref{le:4.6}, we have $0 \leq l \leq 3$ and 
$$\begin{cases} 0<t_0\leq 2,~h_0=1\\2<t_1\leq 5,~h_1=2\\5<t_2\leq 7,~h_2=3\\7<t_3\leq 10,~h_3=4\\\end{cases}.$$ We can obtain

\begin{equation*}
	\begin{split}
		T_{1}=&\bigcup_{\substack{1\leq j \leq 4-\alpha,\\~0\leq i\leq \alpha}}C_{s+2(iq+j)}
		\cup\bigcup_{\substack{106+\alpha \leq j \leq q,\\~0\leq i\leq \alpha-1}}C_{s+2(iq+j)}\\
		&~~\cup\bigcup_{\substack{0 \leq l \leq 3}}~~\cup\bigcup_{\substack{4(2t_l-1)+h_l+ \alpha\leq j\leq 4(2t_l+1)+h_l-1 -\alpha,\\~0 \leq i\leq\alpha-1}}C_{s+2(iq+j)}\\
	\end{split}
\end{equation*}
where $\alpha=1, 2,3,4$ respectively. It is easy to check that $T_1\cap-qT_1=\emptyset$.
Then according to Theorem \ref{th:4.8}, we can obtain the following EAQECCs:\\ $$[[914,468,260;72]]_{109}, [[914,208,478;248]]_{109},$$$$[[914,52,696;528]]_{109},[[914,0,914;912]]_{109}.$$ Especially, they are EAQMDS codes when $d\leq 458$.

\begin{center}
	\begin{sidewaystable}
		$\mathrm{Table~4:}~\mathrm{Some~EAQECCs ~obtained~from~Theorems~\ref{th:4.5}~and ~\ref{th:4.8}}$ 
		\centering 
		\begin{tabular}{|c|c|c|c|c|c|c|}
			\hline
			$m$ &  $a$ &$n$ & $q$ & $\xi $& $q(\xi )$ & $[[n, k, d; c]]_q$\\
			\hline
			&   & &  & $1$& $3$& $[[10, 1, 10; 9]]_3$\\
			&   & &  & $2$ & $5$& $[[26, 4, 16; 8]]_5$, $[[26, 0, 26; 24]]_5$\\
			&   & &  & $3$ & $7$& $[[50, 17, 22; 9]]_7$, $[[50, 5, 36; 25]]_7$, $[[50, 1, 50; 49]]_7$ \\
			$1$ & $2$  &  $q^2+1$ & $2\xi +1$ & $4$ & $9$ & $[[82, 36, 28; 8]]_9$, $[[82, 16, 46; 24]]_9$, $[[82, 4, 64; 48]]_9$, $[[82, 0, 82; 80]]_9$ \\
			&   & &  & $5$& $11$& $[[122, 65, 34; 9]]_{11}$, $[[122, 37, 56; 25]]_{11}$, $[[122, 17, 78; 49]]_{11}$\\
			&   & &  &     & &   $[[122, 5, 100; 81]]_{11}$, $[[122, 1, 122;121]]_{11}$ \\
			&   & &  & $6$ & $13$& $[[170, 100, 40; 8]]_{13}$, $[[170,64, 66; 24]]_{13}$, $[[170, 36, 92; 48]]_{13}$\\
			&   & &  &     &  &   $[[170, 16, 118; 80]]_{13}$, $[[170,4, 144; 120]]_{13}$, $[[170, 0, 170; 168]]_{13}$ \\
			
			\hline
			&   & &  & $1$& $17$& $[[58, 0, 58; 56]]_{17}$\\
			&   & &  $10\xi +7$  & $3$& $37$& $[[274,80,126;56]]_{37}$, $[[274,20,200;144]]_{37}$, $[[274	,	0	,	274	;	272]]_{37}$\\
			
			&   & &  & $4$& $47$& $[[442	,	181	,	160	;	57]]_{47}$, $[[442	,	81	,	254	;	145]]_{47}$, $[[442	,	21	,	348	;	273]]_{47}$, $[[442	,	1	,	442	;	441]]_{47}$\\
			\cline{4-7}
			$3$ &  $10$ &$\frac{q^2+1}{5}$ &  & $1$& $13$& $[[34	,	0	,	34	;	32]]_{13}$\\
			&   & & $10\xi +3$ & $2$& $23$& $[[106	,	21	,	60	;	33]]_{23}$,  $[[106	,	1	,	106	;	105]]_{23}$\\
			&   & &  & $4$& $43$&  $[[370	,	181	,	112	;	33]]_{43}$,  $[[370	,	81	,	198	;	105]]_{43}$,  $[[370	,	21	,	284	;	217]]_{43}$, $[[370	,	1	,	370	;	369]]_{43}$\\
			\cline{4-7}
			\hline
			
			&   & &  & $1$& $47$& $[[170	,	1	,	170	;	169]]_{47}$\\
			&   & &  $26\xi +21$ & $2$& $73$&  $[[410	,	52	,	264	;	168	]]_{73}$, $[[410	,	0	,	410	;	408]]_{73}$ \\
			
			&   & &  & $4$& $125$& $[[1202	,	468	,	452	;	168]]_{125}$, $[[1202	,	208	,	702	;	408]]_{125}$, $[[1202	,	52	,	952	;	752]]_{125}$, $[[1202	,	0	,	1202	;	1200]]_{125}$\\
			\cline{4-7}
			$5$ &  $26$ &$\frac{q^2+1}{13}$ &  & $1$& $31$& $[[74	,	1	,	74	;	73]]_{31}$\\
			&   & & $26\xi +5$ & $3$& $83$& $[[530	,	209	,	198	;	73]]_{83}$, $[[530	,	53	,	364	;	249]]_{83}$, $[[530	,	1	,	530	;	529]]_{83}$\\
			&   & &  & $4$& $109$&  $[[914	,	468	,	260	;	72]]_{109}$,  $[[914	,	208	,	478	;	248]]_{109}$,  $[[914	,	52	,	696	;	528]]_{109}$, $[[914	,	0	,	914	;	912]]_{109}$\\
			\cline{4-7}
			\hline
			
			&   & & $50\xi +43$ & $2$& $143$&  $[[818	,	101	,	532	;	345]]_{143}$, $[[818	,	1	,	818	;	817]]_{143}$ \\
			&   & &    & $3$& $193$& $[[1490	,	400	,	718	;	344]]_{193}$, $[[1490	,	100	,	1104	;	816]]_{193}$, $[[1490	,	0	,	1490	;	1488]]_{193}$\\
			
			\cline{4-7}
			
			$7$ &  $50$ &$\frac{q^2+1}{25}$ &  $50\xi +7$ & $2$& $107$& $[[458	,	101	,	244	;	129]]_{107}$,  $[[458	,	1	,	458	;	457]]_{107}$\\
			&   & &  & $3$& $157$& $[[986	,	400	,	358	;	128]]_{157}$, $[[986	,	100	,	672	;	456
			]]_{157}$, $[[986	,	0	,	986	;	984]]_{157}$\\
			\cline{4-7}
			\hline
			
		\end{tabular}
	\end{sidewaystable}
\end{center}

\section{Conclusion}
Recently, the construction of EAQECCs has attracted so much attention and abundant results have been made by employing classical cyclic, negacyclic and constacyclic codes as code resources. In this paper, a family of EAQECCs of length $n=\frac{2(q^2+1)}{a}$ with general parameters from negacyclic codes is constructed, where $q$ is an odd prime power, $a=m^2+1$ and $m$ is an odd integer. Compared with known parameters in the literature, our EAQECCs can detect and correct more quantum errors thanks to their larger minimum distance. It's worth mentioning that we give a unified rule to construct EAQECCs with form $n=\frac{2(q^2+1)}{a}$, where $a$ and $q$ are defined as before. And we produce some new EAQMDS codes in the meanwhile.

\dse{Acknowledgments}
This research is supported by the National Natural Science Foundation of China (Nos. 12001002, U21A20428, 
12171134, 61972126) and the Natural Science Foundation of Anhui Province (Nos. 2008085QA04, 2108085QA03).

\appendix
	\section{Appendix: The proof of Lemma 4.3}
	\begin{proof}
	(1) For a positive integer $\alpha$ with $1 \leq \alpha \leq \xi $, let 
\iffalse
\begin{equation*}
	\begin{split}
		T_{1}~~&=\bigcup_{\substack{0 \leq l \leq m-1}}\bigcup_{\substack{2t_l\xi+h_l+ \alpha\leq j\leq 2(t_l+1)\xi+h_l-1 -\alpha,\\a_{1l}\lceil \frac{m}{2} \rceil+b_{1l}\lfloor \frac{m}{2} \rfloor\leq t_l \leq a_{2l}\lceil \frac{m}{2} \rceil+b_{2l}\lfloor \frac{m}{2} \rfloor,~h_l=2t_l-l+1,\\if~  j\leq (a-m)\xi+(a-2m),~0 \leq i\leq\alpha,~else,~0 \leq i\leq\alpha-1}}C_{s+2(iq+j)}\\
	\end{split}
\end{equation*}
where $a_{10}=a_{20}=b_{10}=0$, $b_{20}=1$, and \\
$\begin{cases}a_{1l}=a_{1l-1}+1\\b_{1l}=b_{1l-1}\\a_{2l}=a_{2l-1}\\b_{2l}=b_{2l-1}+1\end{cases}$, where $l$ is odd, or
$\begin{cases}a_{1l}=a_{1l-1}\\b_{1l}=b_{1l-1}+1\\a_{2l}=a_{2l-1}+1\\b_{2l}=b_{2l-1}\end{cases}$, where $l$ is even.\\
\fi

\begin{equation*}
	\begin{split}
		T_{1}~~&=\bigcup_{\substack{2t_l\xi+h_l+ \alpha\leq j\leq 2(t_l+1)\xi+h_l-1 -\alpha,\\if~  j\leq (a-m)\xi+(a-2m),~0 \leq i\leq\alpha,~else,~0 \leq i\leq\alpha-1}}C_{s+2(iq+j)}\\
	\end{split}
\end{equation*}
where $0 \leq l \leq m-1$, $a_{1l}\lceil \frac{m}{2} \rceil+b_{1l}\lfloor \frac{m}{2} \rfloor\leq t_l \leq a_{2l}\lceil \frac{m}{2} \rceil+b_{2l}\lfloor \frac{m}{2} \rfloor,~h_l=2t_l-l+1$, $a_{10}=a_{20}=b_{10}=0$, $b_{20}=1$, and  $a_{1l}=a_{1l-1}+l ~{\rm mod} ~2 ,b_{1l}=b_{1l-1}+(l+1) ~{\rm mod} ~2,a_{2l}=a_{2l-1}+(l+1) ~{\rm mod} ~2,b_{2l}=b_{2l-1}+l ~{\rm mod} ~2$.
\
That is
\begin{equation*}
	\begin{split}
		T_{1}&=\bigcup_{\substack{2t_0\xi+2t_0+1+ \alpha\leq j\leq 2(t_0+1)\xi+2t_0 -\alpha,\\0\leq t_0\leq \lfloor \frac{m}{2} \rfloor,~0\leq i\leq \alpha}}C_{s+2(iq+j)}
		~~~\bigcup_{\substack{2t_1\xi+2t_1+ \alpha\leq j\leq 2(t_1+1)\xi+2t_1-1 -\alpha,\\\lceil \frac{m}{2} \rceil\leq t_1\leq 2\lfloor \frac{m}{2} \rfloor,~0\leq i\leq \alpha}}C_{s+2(iq+j)}\\
		&~~\bigcup_{\substack{2t_2\xi+2t_2-1+ \alpha\leq j\leq 2(t_2+1)\xi+2t_2-2 -\alpha,\\\lceil \frac{m}{2} \rceil+\lfloor \frac{m}{2} \rfloor\leq t_2\leq \lceil \frac{m}{2} \rceil+2\lfloor \frac{m}{2} \rfloor,~0\leq i\leq \alpha}}C_{s+2(iq+j)}
		\bigcup_{\substack{2t_3\xi+2t_3-2+ \alpha\leq j\leq 2(t_3+1)\xi+2t_3-3 -\alpha,\\2\lceil \frac{m}{2} \rceil+\lfloor \frac{m}{2} \rfloor\leq t_3\leq \lceil \frac{m}{2} \rceil+3\lfloor \frac{m}{2} \rfloor,~0\leq i\leq \alpha}}C_{s+2(iq+j)}\\
		&\cdots\cdots\\
		&\bigcup_{\substack{2t_{m-1}\xi+2t_{m-1}-m+2+ \alpha\leq j\leq 2(t_{m-1}+1)\xi+2t_{m-1}-m+1 -\alpha,\\ \frac{m^2-m}{2} \leq t_{m-1}\leq  \frac{m^2-1}{2},~0\leq i\leq \alpha-1}}C_{s+2(iq+j)}\\
	\end{split}
\end{equation*}

Then by Lemma \ref{le:4.2}, we have
\begin{equation*}
	\begin{split}
		-qT_{1}&=\bigcup_{\substack{2t_0\xi+2t_0+1+ \alpha\leq j\leq 2(t_0+1)\xi+2t_0 -\alpha,\\0\leq t_0\leq \lfloor \frac{m}{2} \rfloor,~0\leq i\leq \alpha}}C_{s+2(jq-i)}
		~~~\bigcup_{\substack{2t_1\xi+2t_1+ \alpha\leq j\leq 2(t_1+1)\xi+2t_1-1 -\alpha,\\\lceil \frac{m}{2} \rceil\leq t_1\leq 2\lfloor \frac{m}{2} \rfloor,~0\leq i\leq \alpha}}C_{s+2(jq-i)}\\
		&~~\bigcup_{\substack{2t_2\xi+2t_2-1+ \alpha\leq j\leq 2(t_2+1)\xi+2t_2-2 -\alpha,\\\lceil \frac{m}{2} \rceil+\lfloor \frac{m}{2} \rfloor\leq t_2\leq \lceil \frac{m}{2} \rceil+2\lfloor \frac{m}{2} \rfloor,~0\leq i\leq \alpha}}C_{s+2(jq-i)}
		\bigcup_{\substack{2t_3\xi+2t_3-2+ \alpha\leq j\leq 2(t_3+1)\xi+2t_3-3 -\alpha,\\2\lceil \frac{m}{2} \rceil+\lfloor \frac{m}{2} \rfloor\leq t_3\leq \lceil \frac{m}{2} \rceil+3\lfloor \frac{m}{2} \rfloor,~0\leq i\leq \alpha}}C_{s+2(jq-i)}\\
		&\cdots\cdots\\
		&\bigcup_{\substack{2t_{m-1}\xi+2t_{m-1}-m+2+ \alpha\leq j\leq 2(t_{m-1}+1)\xi+2t_{m-1}-m+1 -\alpha,\\ \frac{m^2-m}{2} \leq t_{m-1}\leq  \frac{m^2-1}{2},~0\leq i\leq \alpha-1}}C_{s+2(jq-i)}\\
	\end{split}
\end{equation*}

When $0\leq t_0\leq \lfloor \frac{m}{2} \rfloor$ and $0\leq i\leq \alpha$, then\\
$$\begin{cases}1+\alpha \leq j_{01}\leq 2\xi-\alpha,\\2\xi+3+\alpha \leq j_{02}\leq 4\xi+2-\alpha,\\4\xi+5+\alpha \leq j_{03}\leq 6\xi+4-\alpha,\\\cdots\\(m-1)\xi+m+\alpha \leq j_{0\frac{m-1}{2}}\leq (m+1)\xi+m-1-\alpha. \end{cases}$$
And the following results follow.

When $1+\alpha \leq j_{01}\leq 2\xi-\alpha$, we have 
$$i_{01}q+j_{01}\leq \alpha q+2\xi-\alpha,~(2\xi-\alpha)q-\alpha \leq j_{01}q-i_{01}.$$

When $2\xi+3+\alpha \leq j_{02}\leq 4\xi+2-\alpha$, we have 
$$i_{02}q+j_{02}\leq \alpha q+4\xi+2-\alpha,~(4\xi+2-\alpha)q-\alpha \leq j_{02}q-i_{02}.$$

When $4\xi+5+\alpha \leq j_{03}\leq 6\xi+4-\alpha$, we have 
$$i_{03}q+j_{03}\leq \alpha q+6\xi+4-\alpha,~(6\xi+4-\alpha)q-\alpha \leq j_{03}q-i_{03}.$$
$$\cdots\cdots$$

When $(m-1)\xi+m+\alpha \leq j_{0\frac{m-1}{2}}\leq (m+1)\xi+m-1-\alpha$, we have 
$$i_{0\frac{m-1}{2}}q+j_{0\frac{m-1}{2}}\leq \alpha q+(m+1)\xi+m-1-\alpha,~[(m+1)\xi+m-1-\alpha]q-\alpha \leq j_{0\frac{m-1}{2}}q-i_{0\frac{m-1}{2}}.$$

We have similar results for $t_0,~t_1,~t_2,~\cdots,~t_{\frac{m-1}{2}}$ in all of these cases.

It is easy to check that 
$$s+2(i_{01}q+j_{01})<s+2(j_{01}q-i_{01}),~s+2(i_{01}q+j_{01})<s+2(j_{02}q-i_{02}),$$
$$s+2(i_{01}q+j_{01})<s+2(j_{03}q-i_{03}),~s+2(i_{01}q+j_{01})<s+2(j_{04}q-i_{04}),$$
$$\cdots\cdots$$
$$s+2(i_{01}q+j_{01})<s+2(j_{0\frac{m-1}{2}}q-i_{0\frac{m-1}{2}}),$$
$$\cdots\cdots$$
$$s+2(i_{01}q+j_{01})<s+2(j_{1\frac{m-3}{2}}q-i_{1\frac{m-3}{2}}),$$
$$\cdots\cdots$$
$$s+2(i_{01}q+j_{01})<s+2(j_{(m-1)\frac{m-1}{2}}q-i_{(m-1)\frac{m-1}{2}}).$$
In a similar way, we have 
$$s+2(i_{0\frac{m-1}{2}}q+j_{0\frac{m-1}{2}})<s+2(j_{01}q-i_{01}),$$
$$\cdots$$ $$s+2(i_{0\frac{m-1}{2}}q+j_{0\frac{m-1}{2}})<s+2(j_{(m-1)\frac{m-1}{2}}q-i_{(m-1)\frac{m-1}{2}})$$
$$\cdots$$
$$s+2(i_{1\frac{m-3}{2}}q+j_{1\frac{m-3}{2}})<s+2(j_{01}q-i_{01}),$$
$$\cdots$$ $$s+2(i_{1\frac{m-3}{2}}q+j_{1\frac{m-3}{2}})<s+2(j_{(m-1)\frac{m-1}{2}}q-i_{(m-1)\frac{m-1}{2}})$$
$$\cdots$$
$$s+2(i_{(m-1)\frac{m-1}{2}}q+j_{(m-1)\frac{m-1}{2}})<s+2(j_{01}q-i_{01}),$$
$$\cdots$$ $$s+2(i_{(m-1)\frac{m-1}{2}}q+j_{(m-1)\frac{m-1}{2}})<s+2(j_{(m-1)\frac{m-1}{2}}q-i_{(m-1)\frac{m-1}{2}}).$$

Note that
$iq+j\leq (a-1)\xi +a-2m^2+ 4m-3$, the subscripts of $C_{s+r(u_iq+v_i)}$ is the biggest number in the set. Then $T_1\cap-qT_1=\emptyset$. The desired results follow.	
	
	\end{proof}

	\end{document}